\documentclass[%
reprint,superscriptaddress,
amsmath,amssymb,aps,prl
]{revtex4-1}

\usepackage{import}
\usepackage{amsthm}
\usepackage{dcolumn}
\usepackage{multirow}

\usepackage{graphicx}
\usepackage{dcolumn}
\usepackage{bm}

\usepackage{xcolor}

\newtheorem{theorem}{Theorem}

\theoremstyle{notation}

\theoremstyle{remark}

\renewcommand{\v}{\mathbf{v}}

\begin{document}

\preprint{APS/123-QED}

\title{
Locked entropy in partially coherent optical fields}

\author{Mitchell Harling}
\affiliation{PROBE Lab, School of Engineering, Brown University, Providence, RI 02912, USA}%
\author{Varun A. Kelkar}
\affiliation{Department of Electrical and Computer Eng., University of Illinois at Urbana-Champaign, Urbana, IL 61801, USA}%
\affiliation{Currently at Algorithmic Systems Group, Analog Garage, Analog Devices, Inc., Boston, MA 02110, USA}
\author{Kimani C. Toussaint, Jr.}
\affiliation{PROBE Lab, School of Engineering, Brown University, Providence, RI 02912, USA}%
\author{Ayman F. Abouraddy}%
\affiliation{CREOL, The College of Optics \& Photonics, University of Central Florida, Orlando, FL 32816, USA}%

\date{\today}

\begin{abstract}
We introduce a taxonomy for partially coherent optical fields spanning multiple degrees of freedom (DoFs) based on the rank of the associated coherence matrix (the number of non-zero eigenvalues). When DoFs comprise two spatial modes and polarization, a fourfold classification emerges, with rank-1 fields corresponding to fully coherent fields. We demonstrate theoretically and confirm experimentally that these classes have heretofore unrecognized different properties. Specifically, whereas rank-2 fields can always be rendered separable with respect to its DoFs via a unitary transformation, rank-3 fields are always non-separable. Consequently, the entropy for a rank-2 field can always be concentrated into a single DoF (thus ridding the other DoF of statistical fluctuations), whereas some entropy is always `locked' in one DoF of a rank-3 field.
\end{abstract}

\maketitle

The study of optical coherence and the statistical fluctuations in optical fields extends back to the pioneering work of Zernike \cite{Zernike38P}, and subsequently reached maturation in the work of Wolf and Mandel \cite{Mandel65RMP,Mandel95Book,Wolf07Book}. Recently, new insights into optical coherence have been brought to light \cite{Qian11OL,Kagalwala13NP} by exploiting the mathematical correspondence between the coherence matrix for classical optical fields involving multiple degrees of freedom (DoFs) \cite{Gori98OL,Gori06OL,Gori07OL} and the density operator representing multipartite quantum mechanical states. This correspondence has led to the coinage of the term `classical entanglement' \cite{Kagalwala13NP,Spreeuw1998FP,Spreeuw2001PRA,Aiello2015NJP,FORBES201999,Konrad19CP} to describe optical fields that are \textit{not} separable with respect to their DoFs, in analogy with quantum entanglement that is intrinsic to non-separable multipartite quantum states. The concept of classical entanglement has helped solve problems with regards to Mueller matrices \cite{Simon10PRL}, determine the maximum achievable Young double-slit interference visibility \cite{Abouraddy17OE}, and enable the characterization of quantum optical communications channels \cite{Ndagano17NP}, among many other applications \cite{Otte18LSA,Mamani18JB,kondakci2019classical,Toninelli19AOP,YaoLi20APL,Shen21LSA,Shen21PRR,Hall22JOSAA,Aiello22NJP}.

The study of classical entanglement in optical fields is enriched by the possibility of implementing inter-DoF (or global) unitary transformations (`unitaries' henceforth for brevity \cite{Abouraddy17OE,Halder21OL}), including entangling and disentangling unitaries; e.g., a spatial light modulator can entangle or disentangle polarization and spatial modes \cite{Kagalwala17NC}. This feature is central to the recent demonstration of entropy swapping \cite{Okoro17Optica,harling2022reversible,Harling23JO}, which refers to the reversible reallocation of statistical fluctuations from one DoF to another in a partially coherent field. For example, starting with a \textit{polarized} but spatially \textit{incoherent} field (the entropy is confined to the spatial DoF), a global unitary can convert the field to one that is \textit{unpolarized} but spatially \textit{coherent} (the entropy has been swapped to the polarization DoF with no loss of energy). A similar approach enables entropy concentration, whereby the entropy shared among the DoFs can be optimally transferred into a single DoF via a unitary \cite{harling2022reversible}.

Here we uncover a surprising feature of partially coherent optical fields that places a constraint on entropy concentration under arbitrary global unitaries \cite{Abouraddy17OE,Okoro17Optica}. For concreteness, we examine a canonical optical field model having two binary DoFs, and introduce a fourfold taxonomy based on the \textit{coherence rank} of the associated $4\!\times\!4$ coherence matrix, which corresponds to the number of its non-zero eigenvalues (from 1 to 4). While the rank-1 class embraces all coherent fields, rank-2 through rank-4 classes comprise partially coherent fields. We find that fields of different ranks have altogether different characteristics that have not been investigated to date. Specifically, we find that the potential for concentrating the field entropy into a single DoF depends crucially on the rank. Most conspicuously, the entropy of rank-2 fields -- no matter how \textit{high} -- can \textit{always} be concentrated into a single DoF, thereby leaving the other DoF free of statistical fluctuations [Fig.~\ref{fig:Conversion_concept}]. Indeed, there always exists a global unitary that renders the field separable with respect to its DoFs, with all the initial entropy concentrated into a single DoF. In stark contrast, it is \textit{impossible} to concentrate all the entropy of rank-3 fields -- no matter how \textit{low} -- into one DoF, and residual fluctuations must be retained by the other DoF, which we call `locked entropy' [Fig.~\ref{fig:Conversion_concept}]. This stems from the fact that rank-3 fields possess a \textit{fundamentally non-separable} structure that cannot be eliminated unitarily. We demonstrate these effects experimentally using optical fields defined by polarization and two spatial modes as the binary DoFs of interest. These results open a new window onto understanding the dynamics of optical coherence upon traversing optical systems or media that couple multiple DoFs, and suggests new applications that may exploit the coherence rank in optical imaging and communications.

\textbf{Vector-space formulation of partially coherent optical fields.} The most general state of an optical field characterized by a binary DoF is described by a $2\times2$ coherence matrix. The polarization coherence matrix is $\mathbf{G}_{\mathrm{p}}\!=\!\left(\!\begin{array}{cc}
G^{\mathrm{HH}}&G^{\mathrm{HV}}\\
G^{\mathrm{VH}}&G^{\mathrm{VV}}\end{array}\!\right)$, where $G^{ij}\!=\!\langle E^{i}(E^{j})^{*}\rangle$, $i,j\!=\!\mathrm{H},\mathrm{V}$, and $E^{i}$ is a scalar field component at a point. Similarly, the spatial coherence matrix at two points $a$ and $b$ in a scalar field is
$\mathbf{G}_{\mathrm{s}}\!=\!\left(\!\begin{array}{cc}
G_{aa}&G_{ab}\\G_{ba}&G_{bb}\end{array}\!\right)$, where $G_{kl}\!=\!\langle E_{k}E_{l}^{*}\rangle$, $k,l\!=\!a,b$, and $E_{k}$ is the scalar field at a point. The polarization entropy is $S_{\mathrm{p}}\!=\!-\lambda_{1}\log_{2}\lambda_{1}-\lambda_{2}\log_{2}\lambda_{2}$, where $\lambda_{1}$ and $\lambda_{2}$ are the eigenvalues of $\mathbf{G}_{\mathrm{p}}$; the spatial entropy $S_{\mathrm{s}}$ associated with $\mathbf{G}_{\mathrm{s}}$ is similarly defined. In general $0\!\leq\!S_{\mathrm{p}},S_{\mathrm{s}}\!\leq\!1$, with $S_{\mathrm{p}},S_{\mathrm{s}}\!=\!0$ in the case of fully coherent fields (no statistical fluctuations) \cite{Brosseau06PO}. The maximum entropy is 1~bit when the field is unpolarized or spatially incoherent $\mathbf{G}_{\mathrm{p}},\mathbf{G}_{\mathrm{s}}\!=\!\tfrac{1}{2}\mathcal{I}$ (where $\mathcal{I}$ is the identity matrix).

Taking \textit{both} DoFs (i.e., two points in a vector field), the first-order coherence is described by a $4\!\times\!4$ coherence matrix $\mathbf{G}$ \cite{Gori06OL,Kagalwala13NP,Abouraddy17OE},
\begin{equation}
\mathbf{G}=\left(\begin{array}{cccc}
G_{aa}^{\mathrm{HH}}&G_{aa}^{\mathrm{HV}}&G_{ab}^{\mathrm{HH}}&G_{ab}^{\mathrm{HV}}\vspace{0.1cm}\\
G_{aa}^{\mathrm{VH}}&G_{aa}^{\mathrm{VV}}&G_{ab}^{\mathrm{VH}}&G_{ab}^{\mathrm{VV}}\vspace{0.1cm}\\
G_{ba}^{\mathrm{HH}}&G_{ba}^{\mathrm{HV}}&G_{bb}^{\mathrm{HH}}&G_{bb}^{\mathrm{HV}}\vspace{0.1cm}\\
G_{ba}^{\mathrm{VH}}&G_{ba}^{\mathrm{VV}}&G_{bb}^{\mathrm{VH}}&G_{bb}^{\mathrm{VV}}
\end{array}\right),
\end{equation}
where $G_{kl}^{ij}\!=\!\langle E_{k}^{i}(E_{l}^{j})^{*}\rangle$, $i,j\!=\!\mathrm{H},\mathrm{V}$, and $k,l\!=\!a,b$. The coherence matrices $\mathbf{G}$, $\mathbf{G}_{\mathrm{s}}$, and $\mathbf{G}_{\mathrm{p}}$ are all Hermitian, positive semi-definite, unity-trace matrices. A $4\!\times\!4$ unitary $\hat{U}$ spanning both DoFs \cite{Abouraddy17OE} diagonalizes $\mathbf{G}$: $\mathbf{G}_{\mathrm{D}}\!=\!\hat{U}\mathbf{G}\hat{U}^{\dag}\!=\!\mathrm{diag}\{\lambda_{1},\lambda_{2},\lambda_{3},\lambda_{4}\}$, with $\sum_{j}\lambda_{j}\!=\!1$, and the field can carry up to 2 bits of entropy $S\!=\!-\sum_{j=1}^{4}\lambda_{j}\log_{2}\lambda_{j}$, where $0\!\leq \!S\!\leq\!2$ and $\mathrm{diag}\{\cdot\}$ refers to a diagonal matrix with the listed elements along the diagonal. If, and only if, $\lambda_{1}\lambda_{4}\!=\!\lambda_{2}\lambda_{3}$ can $\mathbf{G}_{\mathrm{D}}$ be separated into a direct product with respect to the two DoFs, $\mathbf{G}_{\mathrm{D}}\!=\!\mathrm{diag}\{\psi_{a},\psi_{b}\}\otimes\mathrm{diag}\{\gamma^{\mathrm{H}},\gamma^{\mathrm{V}}\}$, where each $2\!\times\!2$ coherence matrix corresponds to one DoF \cite{Abouraddy01PRA}. The condition $\lambda_{1}\lambda_{4}\!=\!\lambda_{2}\lambda_{3}$ therefore delineates optical fields that can -- in principle -- be rendered separable with respect to their DoFs via unitaries.

We introduce the reduced coherence matrices that result from `tracing out' one DoF from $\mathbf{G}$: the reduced spatial coherence matrix $\mathbf{G}^{\mathrm{red.}}_{\mathrm{s}}$ after tracing out polarization, and the reduced polarization coherence matrix $\mathbf{G}^{\mathrm{red.}}_{\mathrm{p}}$ after tracing over space. We define entropies $S_{\mathrm{s}}$ and $S_{\mathrm{p}}$ for $\mathbf{G}^{\mathrm{red.}}_{\mathrm{s}}$ and $\mathbf{G}^{\mathrm{red.}}_{\mathrm{p}}$, respectively; in general, $S\!\leq\!S_{\mathrm{s}}+S_{\mathrm{p}}$, with equality occurring only when the field is separable. Crucially, whereas $S$ is invariant with respect to global unitaries, $S_{\mathrm{s}}$ and $S_{\mathrm{p}}$ are \textit{not}. Indeed, whereas $\mathbf{G}$ suffices to completely identify the field coherence, these reduced matrices do \textit{not} \cite{Kagalwala13NP,Kagalwala15SR,Abouraddy17OE,Okoro17Optica}.

\begin{figure}[t!]
\centering
\includegraphics[width=8.6cm]{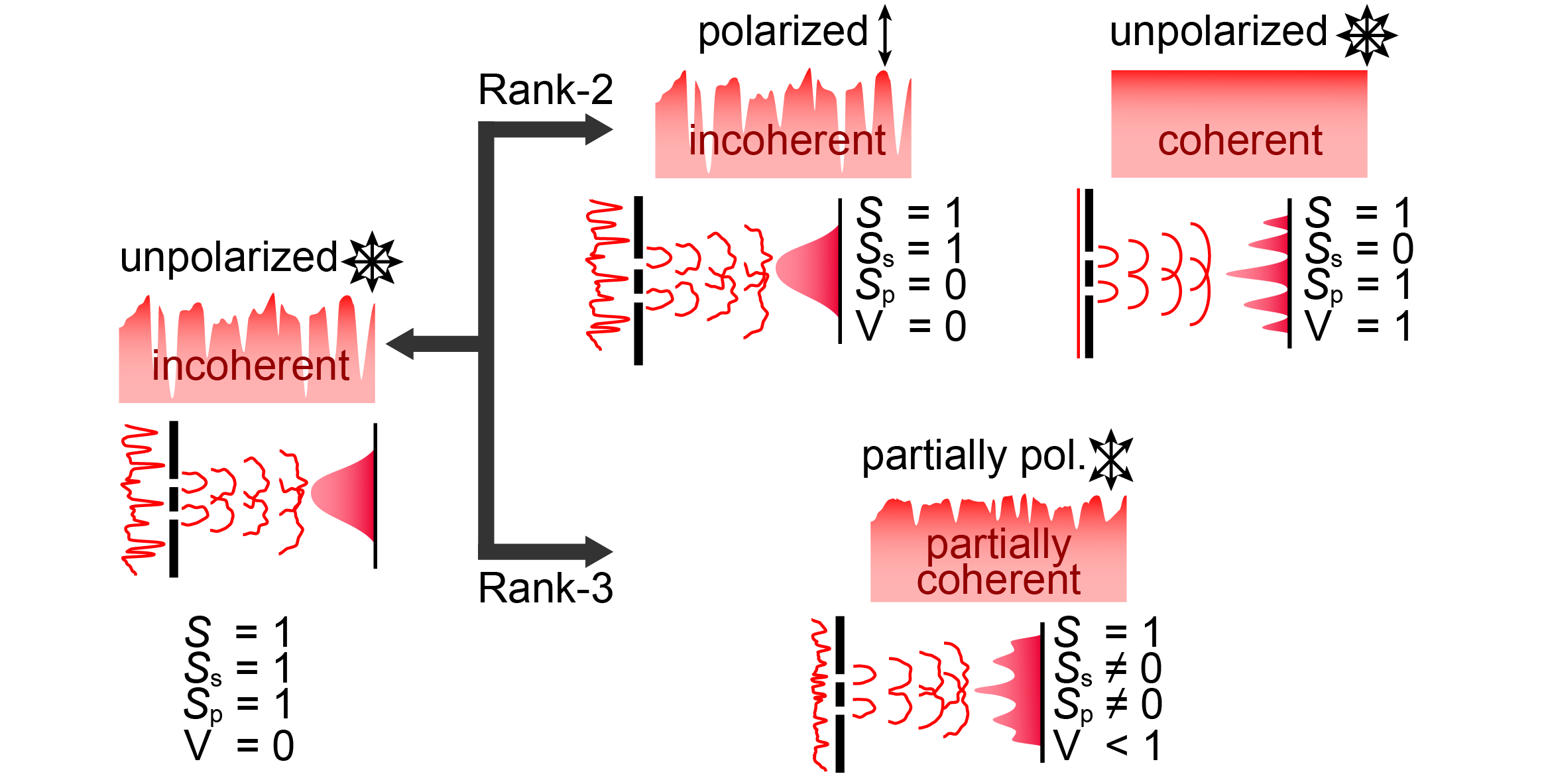}
\caption{Starting with a non-separable field with 1 bit of entropy ($S\!=\!1$, left) that is unpolarized $S_{\mathrm{p}}\!=\!1$ and spatially incoherent $S_{\mathrm{s}}\!=\!1$, a unitary can reversibly convert it into one of two forms depending on the rank of $\mathbf{G}$. For a rank-2 field, the entropy can be \textit{always} fully concentrated into one DoF, leaving the other DoF free of statistical fluctuations. For a rank-3 field, entropy can \textit{never} be fully concentrated in one DoF. There always remains `locked entropy' in the other DoF.}
\label{fig:Conversion_concept}
\end{figure}

\textbf{Coherence rank and entropy concentration.} We classify these optical fields into four families according to their \textit{coherence rank}, defined as the number of non-zero eigenvalues of $\mathbf{G}$. Rank-1 fields, $\{\lambda\}\!=\!\{1,0,0,0\}$, comprise fully coherent fields, $S\!=\!0$ (no statistical fluctuations). It is \textit{always} possible to render rank-1 fields separable via a unitary: $\mathbf{G}\!\rightarrow\!\mathbf{G}_{\mathrm{D}}\!=\!\mathrm{diag}\{1,0\}\otimes\mathrm{diag}\{1,0\}$, whereupon both DoFs are fully coherent.

Partially coherent rank-2 fields, $\{\lambda\}\!=\!\{\lambda_{1},\lambda_{2},0,0\}$, with entropy in the range $0\!<\!S\!\leq\!1$, can \textit{always} be transformed unitarily into the separable form: $\mathbf{G}_{\mathrm{D}}\!=\!\mathrm{diag}\{1,0\}\otimes\mathrm{diag}\{\lambda_{1},\lambda_{2}\}$. This corresponds to a partially polarized field that is fully coherent spatially ($S_{\mathrm{p}}\!=\!S$ and $S_{\mathrm{s}}\!=\!0$). Alternatively, the field can be converted into a fully polarized field that is partially coherent spatially ($S_{\mathrm{p}}\!=\!0$ and $S_{\mathrm{s}}\!=\!S$). In general, the entropy of a rank-2 field is shared between the two DoFs. Nevertheless, even in its highest-entropy state $S\!=\!1$,  $\{\lambda\}\!=\!\{\tfrac{1}{2},\tfrac{1}{2},0,0\}$, such fields can always be rendered separable such that one DoF is fully coherent (ridding it completely from statistical fluctuations), with the 1 bit of field entropy concentrated in the other DoF \cite{Okoro17Optica,harling2022reversible,Harling23JO}; see Fig.~\ref{fig:Conversion_concept}.

In stark contrast, the coherence matrices associated with rank-3 fields, $\{\lambda\}\!=\!\{\lambda_{1},\lambda_{2},\lambda_{3},0\}$, whose entropy is in the range $0\!<\!S\!\leq\!1.585$, can\textit{not} be expressed as a direct product ($\lambda_{1}\lambda_{4}\!=\!0\!\neq\!\lambda_{2}\lambda_{3}$); that is, rank-3 fields are \textit{never} separable with respect to their DoFs. This fundamental non-separability is independent of the values $\{\lambda\}$ and is solely a consequence of the rank of $\mathbf{G}$. This hitherto unrecognized feature has important consequences for entropy concentration: it prevents ridding either DoF altogether from statistical fluctuations. Indeed, after concentrating the entropy into one DoF, a residual amount of entropy is retained that we refer to as \textit{locked entropy}. The entropy in a rank-3 field must always be shared between the DoFs no matter how low $S$ is. Even when $S\!<\!1$, it is impossible to realize the condition $S_{\mathrm{p}}\!=\!S$ and $S_{\mathrm{s}}\!=\!0$ (or $S_{\mathrm{p}}\!=\!0$ and $S_{\mathrm{s}}\!=\!S$) unitarily, which is attainable for rank-2 fields \textit{of the same entropy} [Fig.~\ref{fig:Conversion_concept}]. Furthermore, when $S\!>\!1$ one cannot concentrate 1 bit of entropy in one of the DoFs. Defining the function $f(x)\!=\!-\!x\log_{2}{x}\!-\!(1\!-\!x)\log_{2}{(1\!-\!x)}$, the minimum entropy that is locked in one DoF is $S_{\mathrm{min}}\!=\!f(\lambda_{1}\!+\!\lambda_{2})$, in which case the entropy concentrated into the other DoF is $S_{\mathrm{max}}\!=\!f(\lambda_{1}\!+\!\lambda_{3})$.

Rank-4 fields, $\{\lambda\}\!=\!\{\lambda_{1},\lambda_{2},\lambda_{3},\lambda_{4}\}$, can sometimes be unitarily rendered separable with respect to their DoFs depending on the eigenvalues, and they thus share the properties of rank-2 or rank-3 fields. We will not examine rank-4 fields here and focus instead on delineating the characteristics of rank-2 and rank-3 fields.

\begin{figure}[t!]
\centering
\includegraphics[width=8.6cm]{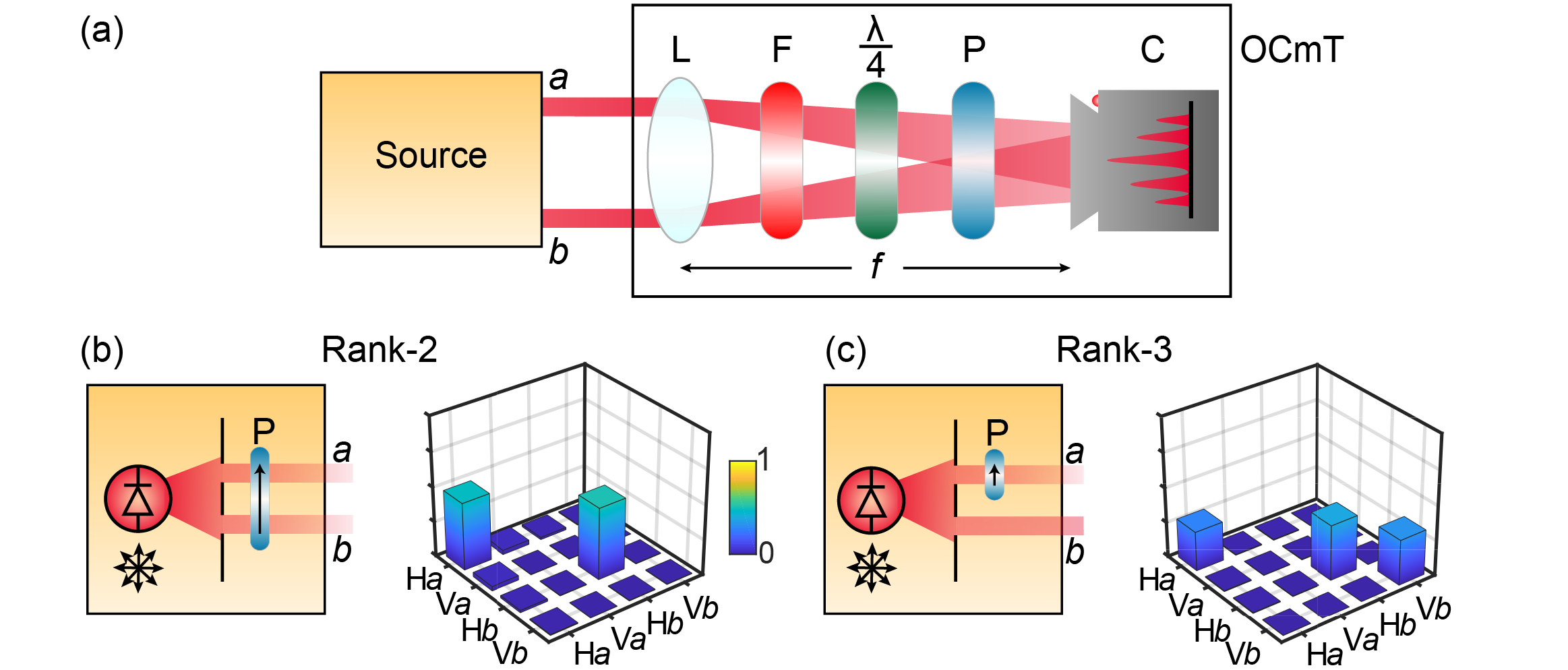}
\caption{(a) Schematic of the OCmT measurement scheme used to measure coherence matrices; L: spherical lens (focal length $f\!=\!30$~cm); F: spectral filter; $\tfrac{\lambda}{4}$: quarter wave plate; P: linear polarizer; C: CMOS camera. (b) Source preparation for a rank-2 field and the measured $4\times4$ coherence matrix $\mathbf{G}$. (c) Same as (b) for a rank-3 field.}
\label{fig:setup}
\end{figure}

\textbf{Experiment.} We first prepare and characterize representative rank-2 and rank-3 fields [Fig.~\ref{fig:setup}]. Starting from unpolarized, spatially incoherent light from an LED (wavelength 625~nm), we select two spatial modes using slits at points $a$ and $b$ that are sufficiently separated to guarantee mutual incoherence [Fig.~\ref{fig:setup}(a)]. For a rank-2 field $\mathbf{G}\!=\!\tfrac{1}{2}\mathrm{diag}\{1,0,1,0\}$, the source configuration along with the measured coherence matrix are shown in Fig.~\ref{fig:setup}(b), and the corresponding results for the rank-3 field with $\mathbf{G}\!=\!\tfrac{1}{3}\mathrm{diag}\{1,0,1,1\}$ are shown in Fig.~\ref{fig:setup}(c). The rank-2 field is prepared by placing a polarizer at both $a$ and $b$, yielding $S\!=\!1$: the field is polarized $S_{\mathrm{p}}\!=\!0$ but spatially incoherent $S_{\mathrm{s}}\!=\!1$. The rank-3 field is prepared by placing a linear polarizer at $b$ only (the field at $a$ remains unpolarized) to yield $S\!=\!1.585$: the field is partially polarized and partially coherent spatially. Throughout, $\mathbf{G}$ is reconstructed via optical coherence matrix tomography (OCmT) [Fig.~\ref{fig:setup}(a)], which extends to optical fields with multiple DoFs \cite{Abouraddy14OL,Kagalwala15SR} the analogous procedure of quantum state tomography \cite{Wooters90CEPI,James01PRA1,Abouraddy02OptComm}; see Supplementary for further experimental details.

\begin{figure}[t!]
\centering
\includegraphics[width=8.6cm]{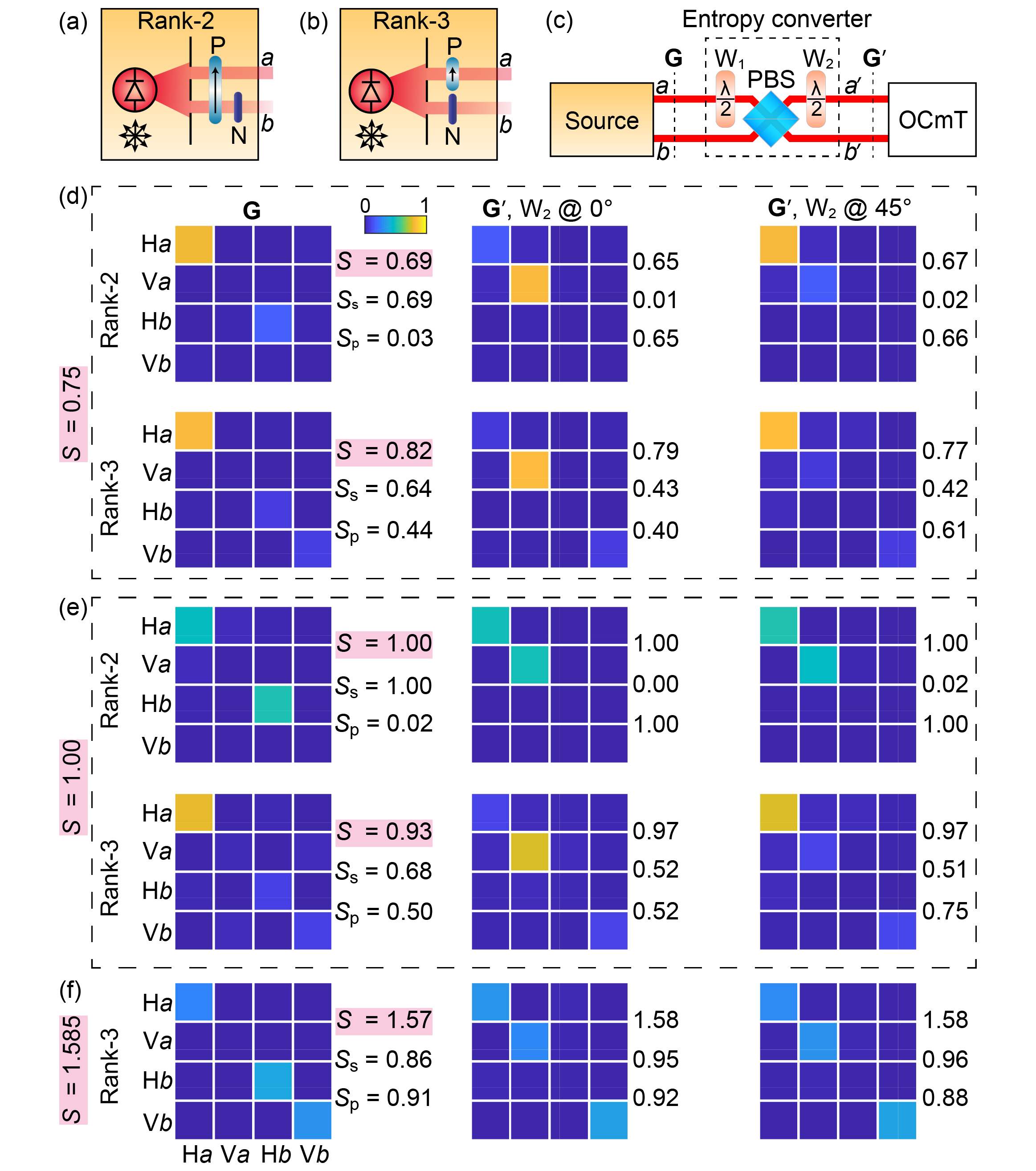}
\caption{Unitary entropy conversion for rank-2 and rank-3 fields. (a) Source configurations for rank-2 and (b) rank-3 fields. P: Linear polarizer oriented along H; N: neutral density filter. (c) Setup for entropy conversion. W: Half-wave plate; PBS: polarizing beam splitter. (d-f) From left to right: $\mathbf{G}$ reconstructed before the entropy converter; $\mathbf{G}'$ after the entropy converter with W$_{2}$ oriented at $0^{\circ}$; and $\mathbf{G}'$ with W$_{2}$ at $45^{\circ}$. All matrices are measurements, and the fidelity throughout was $>\!98\%$ with respect to theoretical expectations (see Supplementary). (d) Rank-2 and rank-3 fields with $S\!\approx\!0.75$; (e) same as (d) for $S\!\approx\!1$; and (f) rank-3 field with $S\!\approx\!1.585$.}
\label{fig:primes}
\end{figure}

The impact of the coherence rank on the limits of entropy concentration is illustrated in Fig.~\ref{fig:primes}. We consider iso-entropy rank-2 [Fig.~\ref{fig:primes}(a)] and rank-3 [Fig.~\ref{fig:primes}(b)] fields. We make use of an entropy converter that unitarily couples the two DoFs [Fig.~\ref{fig:primes}(c)], which comprises a half-wave plate (HWP) W$_{1}$ in path $a$ oriented at $45^{\circ}$ with respect to H (H$\rightarrow$V, V$\rightarrow$H), a polarizing beam splitter (PBS) that couples modes $a$ and $b$ and produces modes $a'$ and $b'$, followed by a HWP W$_{2}$ in mode $a'$ in one of two orientations: at $0^{\circ}$ with H (H$\rightarrow$H and V$\rightarrow-$V), and at $45^{\circ}$ with H (H$\rightarrow$V, V$\rightarrow$H). The first orientation minimizes the entropy in the spatial DoF (entropy concentration), while the second orientation swaps the entropy of the spatial and polarization DoFs (entropy swapping).  

Either binary DoF (polarization or spatial modes) can support up to 1~bit of entropy. We thus first prepare rank-2 and rank-3 fields with $S\!=\!0.75$ [Fig.~\ref{fig:primes}(d)]. For the rank-2 field, the entire entropy can be concentrated in the spatial DoF, $S_{\mathrm{s}}\!=\!0.75$ (partially coherent spatially) and $S_{\mathrm{p}}\!=\!0$ (fully polarized). Using the first setting for W$_{2}$, the entropy converter minimizes the spatial entropy: $S_{\mathrm{s}}\!\rightarrow\!0$ (spatially coherent) and $S_{\mathrm{p}}\!\rightarrow\!0.75$ (partially polarized). The second setting for W$_{2}$ swaps the entropy between the DoFs, which yields here the same result as that of entropy concentration with the first setting.

The corresponding results for the rank-3 field are entirely in contrast to those for the iso-entropy $S\!=\!0.75$ rank-2 field. The rank-3 source configuration yields theoretical values of $S_{\mathrm{s}}\!=\!0.6$ (partially coherent spatially) and $S_{\mathrm{p}}\!=\!0.38$ (partially polarized); see Supplementary. The first setting minimizes the spatial entropy but can\textit{not} concentrate all the entropy into the polarization DoF; rather, some entropy remains locked in the spatial DoF $S_{\mathrm{s}}\!\rightarrow\!0.38$. The second setting for the entropy converter swaps the entropy between the DoFs: $S_{\mathrm{s}}\!\rightarrow\!0.38$ (partially coherent spatially) and $S_{\mathrm{p}}\!\rightarrow\!0.6$ (partially polarized). Similar results are obtained when the initial field has a total of 1 bit of entropy, $S\!=\!1$ [Fig.~\ref{fig:primes}(e)]. Whereas the entire entropy can be concentrated in either DoF in the case of a rank-2 field, this cannot be achieved for the iso-entropy rank-3 field, and some entropy must remain locked in either DoF. Finally, the entropy of rank-3 can exceed 1 bit (whereas that of rank-2 fields cannot). In Fig.~\ref{fig:primes}(f) we repeat the measurements with a maximum-entropy rank-3 field, $S\!=\!1.585$. Here the locked entropy in the spatial DoF is $S_{\mathrm{s}}\!=\!f(\tfrac{2}{3})\!=\!0.92$. 

The field rank can be identified by reconstructing $\mathbf{G}$, as shown in Fig.~\ref{fig:primes}. Nevertheless, information concerning the coherence rank can be deduced by observing the visibility of the spatial interference fringes produced by the field when the fields at $a$ and $b$ are superposed after a polarization projection. Two theorems (see Supplementary for proofs) help establish a strategy for this approach.

\begin{theorem}\label{thm:rank3}
For a vector optical field supported on two spatial points with a coherence matrix $\mathbf{G}$, if there exists a polarization projection along vector $\mathbf{P}$ along which the field is spatially coherent (i.e., it can produce spatial interference fringes with $100\%$ visibility), then $\mathcal{R}(\mathbf{G})\leq3$. 
\end{theorem}

\begin{theorem}\label{thm:rank2}
For a vector optical field supported on two points with a coherence matrix $\mathbf{G}$, if there exist two orthogonal polarization projections $\mathbf{P}$ and $\mathbf{Q}$ along which the field is spatially coherent (i.e., it can produce $100\%$-visibility spatial interference fringes), then $\mathcal{R}(\mathbf{G})\!\leq\!2$. 
\end{theorem}

\begin{figure}[t!]
\centering
\includegraphics[width=8.6cm]{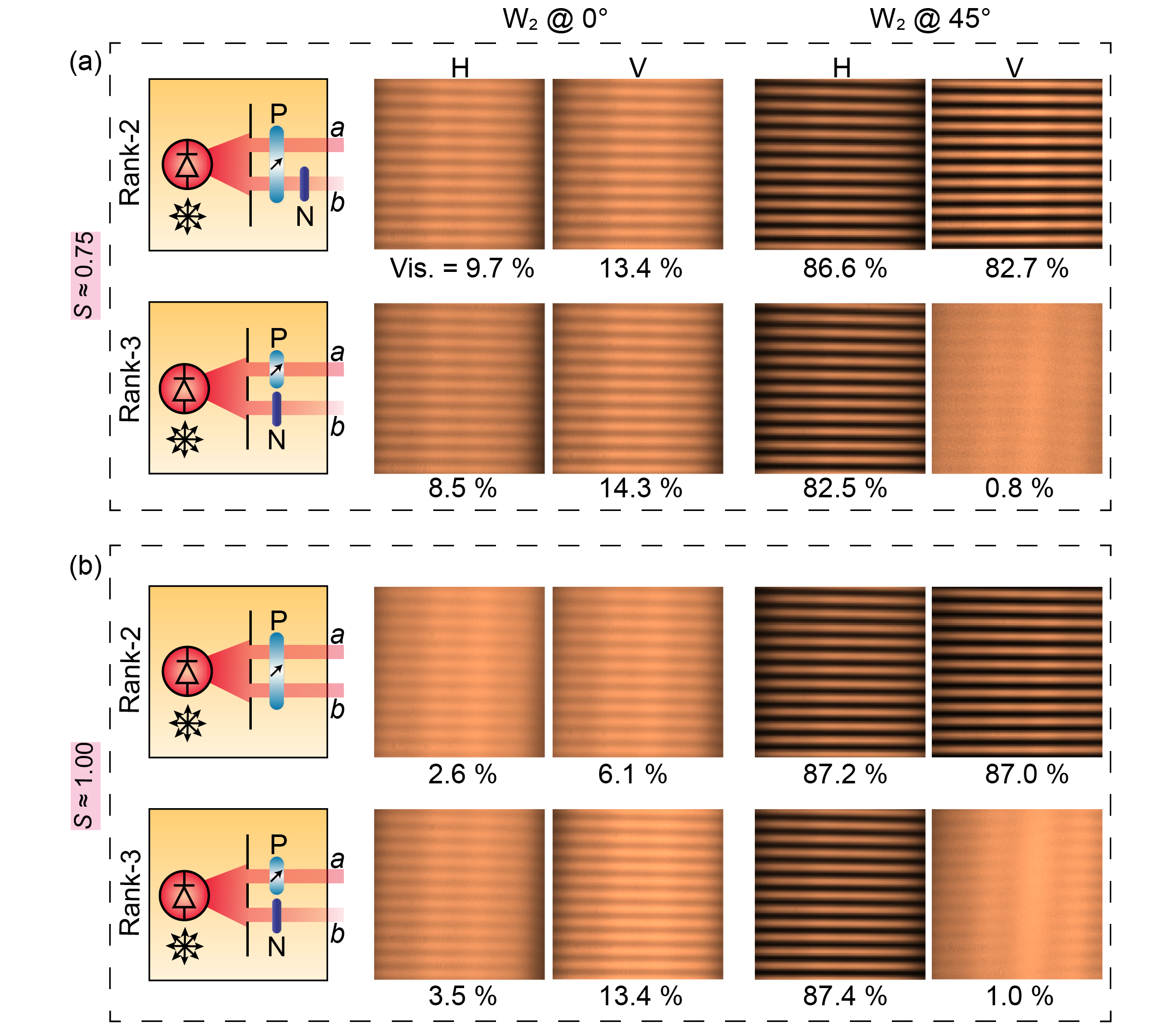}
\caption{Identifying the coherence rank through the spatial coherence after a polarization projection. (a,b) From left to right: the source preparation; optimal interference fringes along the H and V polarization projections after the entropy converter in Fig.~\ref{fig:primes}(c), with W$_2$ oriented at $0^{\circ}$ with H; and optimal interference fringes along the H and V polarization projections with W$_2$ oriented at $45^{\circ}$. (a) Iso-entropy rank-2 and rank-3 fields with $S\!=\!0.75$. (b) Same as (a) for $S\!=\!1$.}
\label{fig:visibility}
\end{figure}

In other words, identifying an orthogonal pair of polarization projections that both yield a spatially coherent field indicates that the field is either rank-1 or rank-2. Identifying only a single polarization projection that yields a spatially coherent field indicates that the field is rank-3. There is \textit{no} polarization projection for a rank-4 field that yields a spatially coherent field.

We demonstrate these results experimentally in Fig.~\ref{fig:visibility} with pairs of iso-entropy rank-2 and rank-3 fields. After the field is prepared, it is directed through the entropy converter shown in Fig.~\ref{fig:primes}(c), and then the field is globally projected onto a prescribed polarization. We search for pairs of directions along which the resulting scalar field yields spatial interference fringes with $100\%$ visibility.

We start with a pair of fields at $S\!\approx\!0.75$ [Fig.~\ref{fig:visibility}(a)]. The rank-2 field is prepared by projecting the polarization at $45^{\circ}$ with respect to H and adjusting the amplitude of one spatial mode to obtain the targeted entropy (see Supplementary for the full coherence matrices associated with the fields in Fig.~\ref{fig:visibility}). After the entropy converter with W$_{2}$ oriented at $0^{\circ}$, no spatial interference fringes of high visibility are observed at any polarization projection. After setting W$_{2}$ at $45^{\circ}$, the polarization projections along H and V yield high-visibility spatial interference fringes, as expected for a rank-2 field.

We contrast these observations with those for an iso-entropy rank-3 field $S\!\approx\!0.75$. This field is prepared by projecting the polarization at $a$ alone along $45^{\circ}$ and adjusting the amplitude at $b$ to obtain the target entropy. After the entropy converter with W$_{2}$ oriented at $0^{\circ}$, no spatial interference fringes are observed at any polarization projection. However, after setting W$_{2}$ at $45^{\circ}$, projecting the polarization along H yields a field that produces high-visibility spatial interference fringes. The corresponding polarization projection along V does \textit{not} yield a spatially coherent field, and no interference fringes can be observed. We increase the field entropy for an iso-entropy pair of rank-2 and rank-3 fields to $S\!\approx\!1$ (the maximum entropy for rank-2) [Fig.~\ref{fig:visibility}(b)], and observe similar results to those for the lower-entropy fields [Fig.~\ref{fig:visibility}(a)]. Despite the higher entropy, we can still identify a pair of polarization projections for the rank-2 field that result in spatial coherence, whereas only a single polarization projection is identified for the rank-3 field.

\noindent\textbf{Discussion.} The approach outlined here in terms of coherence matrices \cite{Fano57RMP,Gamo64PO,Perina85book,Gori06OL} reveals features that are difficult to discern otherwise when extended to multiple DoFs. The analysis and experiments suggest a wealth of fundamental questions regarding the statistical behavior of optical fields: How does the rank vary spatially across a vector optical field? How does the spatial distribution of the rank evolve with free propagation? How is the coherence rank affected by optical nonlinearities? Although we have couched the coherence matrix here in terms of polarization and spatial modes, this description can be extended to other DoFs, including higher-dimensional DoFs (e.g., orbital angular momentum), and even continuous DoFs after implementing the Schmidt decomposition to obtain an effective finite-dimensional representation \cite{Ekert95AJP,Law00PRL,Law04PRL,Eberly06LP,Hall22JOSAA}. This is particularly relevant in light of recent realizations of optical fields in which the spatial, temporal, and polarization DoFs are all coupled \cite{Diouf21OE,yessenov2022space,Yessenov22AOP,Yessenov22OL}. In addition to the intrinsic interest of the coherence rank as a potential thermodynamic variable for electromagnetic fields, it may also serve as an integer identifier of the global properties of the field to be exploited for communications schemes using partially coherent light \cite{Nardi22OL}. 

In conclusion, we have presented a classification scheme of partially coherent optical fields based on the rank of the $4\times4$ coherence matrix for two binary DoFs. This classification unveils surprising structural distinctions: \textit{all} rank-2 fields are fundamentally separable whereas all rank-3 fields are intrinsically \textit{non}-separable. Consequently, the entropy in rank-2 fields -- no matter how high -- can always be concentrated into one DoF, thereby leaving the other DoF free of statistical fluctuations. In contrast, in a rank-3 field the entropy -- no matter how low -- cannot be fully concentrated into one DoF, and locked entropy remains associated with the other DoF.

\vspace{0.5cm}
\noindent
\begin{acknowledgments}
We thank C. Okoro, M. Yessenov, and A. Dogariu for useful discussions and assistance. This work was funded by the US Office of Naval Research (ONR) under contracts N00014-17-1-2458 and N00014-20-1-2789.
\end{acknowledgments}

\bibliography{main}

\clearpage

\end{document}


\title{
Supplementary: Locked entropy in partially coherent optical fields}



\begin{abstract}
We provide details of the experimental setup described in the main text, which is used for demonstrating entropy concentration and thus revealing the phenomenon of `locked entropy' in partially coherent optical fields. We also present all the theoretically expected coherence matrices corresponding to the experimental configurations reported in the main text. Finally, we offer proofs of the two theorems concerning the visibility of double-slit interference for optical fields of different coherence ranks.   
\end{abstract}

\maketitle

\section{Tomographic reconstruction of coherence matrices}

We make use of optical coherence matrix tomography (OCmT) to reconstruct the $4\times4$ coherence matrix $\mathbf{G}$ \cite{Abouraddy14OL,Kagalwala15SR}. The setup is shown in Fig.~2(a) in the main text. The outline of the OCmT procedure is as follows:
\begin{enumerate}
\item \textit{Optical setup.} The source is a light-emitting diode (LED) having a center wavelength of 625~nm, and a bandwidth (FWHM) of $\approx\!17$~nm. The spectral bandpass filter that has a center wavelength of 620~nm and a FWHM of 10~nm serves to improve the temporal coherence of the field by reducing the bandwidth. The field is restricted spatially via two vertically-oriented narrow slits to two spatial modes $a$ and $b$ [Fig.~\ref{fig:OCmT_supplemental}(a)]. The width of each slit is 100~$\mu$m, and they are separated by 23~mm. The separation distance is selected to be significantly larger than the spatial coherence width of the LED. A spherical lens of focal length $f\!=\!30$~mm is used to overlap the spatial modes from $a$ and $b$ at the Fourier plane, whereby interference fringes are observed if the field is spatially coherent. The polarizer provides linear projections along three directions: H (horizontal), V (vertical), and D (diagonal, $45^{\circ}$). In conjunction with a quarter wave plate, we also obtain a projection in the circular polarization basis (R).
 
\item \textit{The intensity measurements acquired for OCmT reconstruction.} The above-described setup provides the following intensity measurements: $I^{j}_{k}$, where $j\!=$ H, V, D, and R represents the polarization projection; and $k\!=\!a,b,$ $a+b$, and $a+ib$ are the spatial projections. The intensity corresponding to the spatial projection $a$ is obtained by blocking the slit $b$ and measuring the peak intensity at the center of the diffraction pattern; and similarly for the spatial projection $b$. The spatial projection $a+b$ is obtained by allowing light to pass through both slits and then measuring the intensity at the center of the interference pattern. The spatial projection $a+ib$ is the intensity at the location midway between the central peak and the first dip (minimum) of the interference pattern. The full list of acquired intensity measurements is provided in Table~\ref{Table:OCmT}.

\begin{table}[b!]
    \centering
    \caption{The tomographic intensity measurements necessary for  OCmT. H: horizontal, V: vertical, D: 45°, R: right-handed circular.}
    \begin{tabular}{|l|*{6}{c|}}
    \hline \multicolumn{1}{|c}{}& &\multicolumn{4}{c|}{\textbf{polarization projection}}\\
    \cline{3-6} \multicolumn{1}{|c}{}&\multicolumn{1}{c|}{} &H & V & D & R\\ 
    \hline \multirow{4}{*}{\centering{\textbf{spatial projection}}} &$a$ & $I^{\mathrm{H}}_{a}$ & $I^{\mathrm{V}}_{a}$ & $I^{\mathrm{D}}_{a}$ &$I^{\mathrm{R}}_{a}$ \\
    \cline{2-6} &$b$ & $I^{\mathrm{H}}_{b}$ & $I^{\mathrm{V}}_{b}$ & $I^{\mathrm{D}}_{b}$ & $I^{\mathrm{R}}_{b}$\\
    \cline{2-6} &$a+b$ & $I^{\mathrm{H}}_{(a+b)}$ & $I^{\mathrm{V}}_{(a+b)}$ & $I^{\mathrm{D}}_{(a+b)}$ & $I^{\mathrm{R}}_{(a+b)}$\\
    \cline{2-6} &$a+ib$ & $I^{\mathrm{H}}_{(a+ib)}$ & $I^{\mathrm{V}}_{(a+ib)}$ & $I^{\mathrm{D}}_{(a+ib)}$ & $I^{\mathrm{R}}_{(a+ib)}$ \\
    \hline
    \end{tabular}
    \label{Table:OCmT}
\end{table}

\item \textit{Calculating the multi-DoF Stokes parameters and reconstructing the coherence matrix.} When reconstructing the coherence matrix in polarimetry, the intensity measurements are used to calculate the polarization Stokes parameters. One may also define generalized Stokes parameters to encompass two DoFs in analogy to the Stokes parameters used to characterize two-photon quantum states \cite{Abouraddy02OptComm}. In terms of the multi-DoF Stokes parameters $S_{lm}$, where $l,m\!=\!0,1,2,3$, the multi-DoF coherence matrix $\mathbf{G}$ is given by:
\begin{equation}
\mathbf{G}=\frac{1}{4}\sum_{l,m=0}^{3} S_{lm} \hat{\sigma}^{p}_{l}\otimes\hat{\sigma}^{s}_{m},    
\end{equation}
where $\hat{\sigma}_{l}$ and $\hat{\sigma}_{m}$ represent the standard Pauli matrices. The coherence matrix $\mathbf{G}$ is thus given explicitly as follows:
\begin{equation}
    \centering    
    \mathbf{G}=
    \left(\begin{array}{cccc}
    S_{00} + S_{01} + S_{10} + S_{11} 
    & S_{02} + S_{12} - i(S_{03} + S_{13}) 
    & S_{20} + S_{21} - i(S_{30} + S_{31}) 
    & S_{22} - S_{33} - i(S_{23} + S_{32}\\
    S_{02} + S_{12} + i(S_{03} + S_{13}) 
    & S_{00} - S_{01} + S_{10} - S_{11} 
    & S_{22} - S_{33} + i(S_{23} - S_{32})
    & S_{20} - S_{21} - i(S_{30} - S_{31})\\
    S_{20} + S_{21} + i(S_{30} + S_{31})  
    & S_{22} + S_{33} - i(S_{23} - S_{32})  
    & S_{00} + S_{01} - S_{10} - S_{11} 
    & S_{02} - S_{12} - i(S_{03} - S_{13})\\
    S_{22} - S_{33} + i(S_{23} + S_{32})  
    & S_{20} - S_{21} + i(S_{30} - S_{31})  
    & S_{02} - S_{12} + i(S_{03} - S_{13})  
    & S_{00} - S_{01} - S_{10} + S_{11} \\
    \end{array}\right).
    \label{eq:G with Stokes}
    \end{equation} 

The multi-DoF Stokes parameters are defined in terms of intensity measurements as follows:
\begin{equation}
S_{lm}=4I_{lm}-2I_{0m}-2I_{l0}+I_{00},
\end{equation}
where the intensity values $I_{lm}$ are defined in terms of the measured intensity projections as follows:
\begin{eqnarray}
I_{00}&=&I^{\mathrm{H}}_{a}+I^{\mathrm{V}}_{a}+I^{\mathrm{H}}_{b}+I^{\mathrm{V}}_{b},\nonumber\\
I_{01}&=&I^{\mathrm{H}}_{a} + I^{\mathrm{H}}_{b},\nonumber\\
I_{02}&=& I^{\mathrm{D}}_{a} + I^{\mathrm{D}}_{b},\nonumber\\
I_{03}&=&I^{\mathrm{R}}_{a} + I^{\mathrm{R}}_{b},\nonumber\\
I_{10}&=&I^{\mathrm{H}}_{a} + I^{\mathrm{V}}_{a},\nonumber\\
I_{11}&=&I^{\mathrm{H}}_{a},\nonumber\\
I_{12}&=&I^{\mathrm{D}}_{a},\nonumber\\
I_{13}&=&I^{\mathrm{R}}_{a},\nonumber\\
I_{20}&=&I^{\mathrm{H}}_{a},\nonumber\\
I_{20}&=&(I^{\mathrm{H}}_{(a+b)}+I^{\mathrm{V}}_{(a+b)})/2,\nonumber\\
I_{21}&=&I^{\mathrm{H}}_{(a+b)}/2,\nonumber\\
I_{22}&=&I^{\mathrm{D}}_{(a+b)}/2,\nonumber\\
I_{23}&=&I^{\mathrm{R}}_{(a+b)}/2,\nonumber\\
I_{30}&=&(I^{\mathrm{H}}_{(a+ib)}+I^{\mathrm{V}}_{(a+ib)})/2,\nonumber\\
I_{31}&=&I^{\mathrm{H}}_{(a+ib)}/2,\nonumber\\
I_{32}&=&I^{\mathrm{D}}_{(a+ib)}/2,\nonumber\\
I_{33}&=&I^{\mathrm{R}}_{(a+ib)}/2.
\end{eqnarray}
\item The experimental matrices shown in Fig.~2 and Fig.~3 in the main text were processed by first diagonalizing the measured matrix. We set a noise threshold for the eigenvalues and eliminated diagonal elements below this threshold ($\lambda\!=\!0.06$). The diagonalized matrix was converted back to a $4\times4$ matrix using the same unitary transformation used in the diagonalization. Each matrix is normalized to its trace, and the entropies and fidelities were then calculated.
\end{enumerate}

\begin{figure}[t!]
\centering
\includegraphics[width=8.6cm]{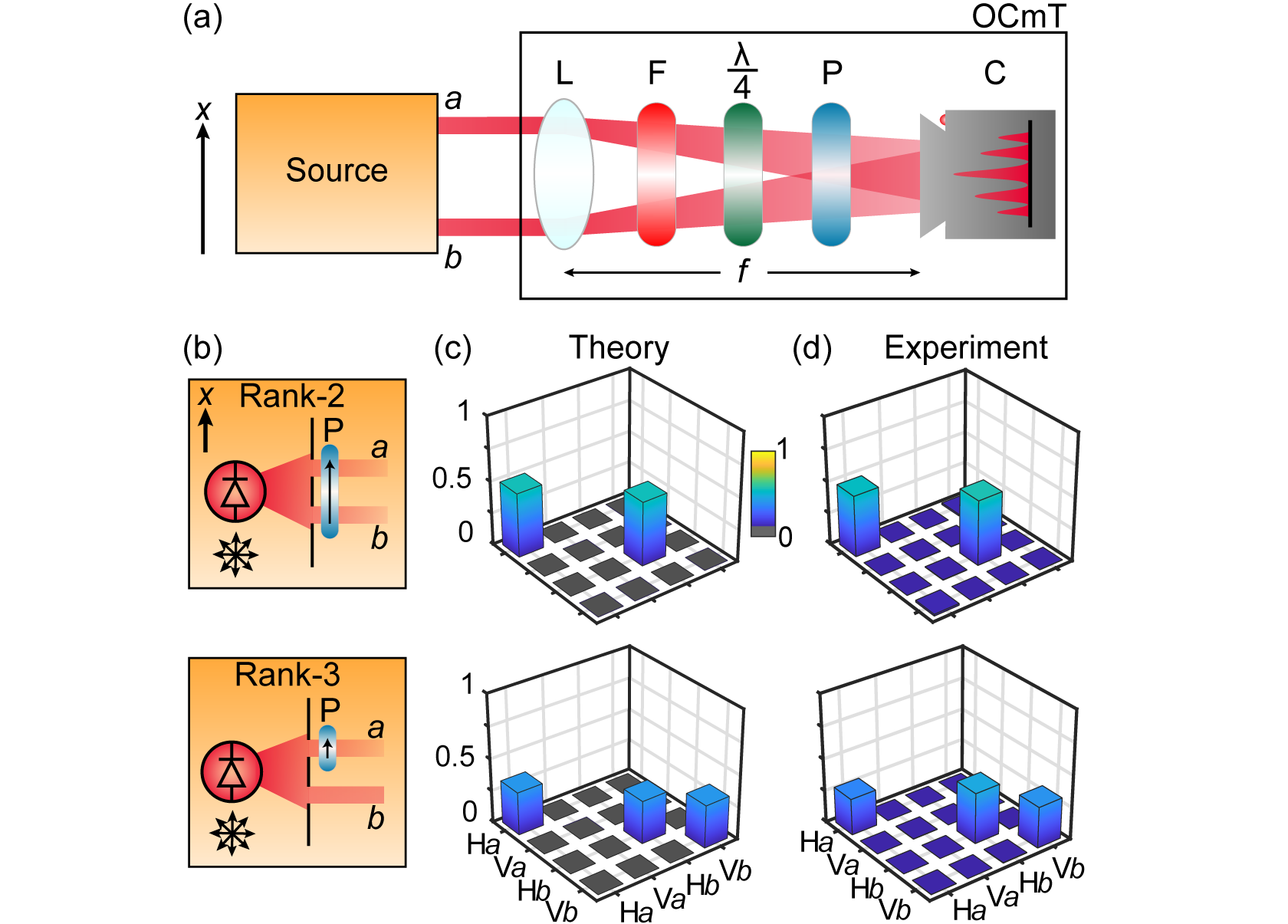}
\caption{Measuring rank-2 and rank-3 fields. (a) Schematic of the OCmT measurement scheme used to measure the coherence matrices; L: spherical lens (focal length $f\!=\!30$~mm); F: spectral filter; $\lambda/4$: quarter wave plate; P: linear polarizer; C: CMOS camera. (b) Source preparation for rank-2 and rank-3 fields. (c) The theoretical $4\times4$ coherence matrices for rank-2 and rank-3 fields. (d) The experimentally reconstructed $4\times4$ coherence matrices for rank-2 and rank-3 fields measured using OCmT.}
\label{fig:OCmT_supplemental}
\end{figure}

\section{Unitary transformation matrix for the entropy converter}

The unitary transformation associated with the entropy converter in Fig.~3(c) of the main text can be derived by cascading those for the half-wave plate (HWP) W$_{1}$ oriented at an angle $\theta_{1}$ with the H polarization, the polarizing beam splitter (PBS), and the HWP W$_{2}$ oriented at $\theta_{2}$. It is critical that the matrices for these unitaries be constructed over the full vector space of spatial modes \textit{and} polarization. The unitary transformation for a HWP oriented at an angle $\theta$ with H is thus given by:  
\begin{equation}\label{Eq:HWP in a}
U_{\mathrm{W}}(\theta)=\left(\begin{array}{cccc}
\cos2\theta&\sin2\theta&0&0\\\sin2\theta&-\cos2\theta&0&0\\
0&0&1&0\\0&0&0&1\end{array}\right),
\end{equation}
and that for the PBS (assumed to reflect the V polarization components) is given by \cite{Abouraddy17OE}:
\begin{equation}\label{Eq:PBS}
U_{\mathrm{PBS}}=\left(\begin{array}{cccc}
0&0&1&0\\0&1&0&0\\1&0&0&0\\0&0&0&1\end{array}\right).
\end{equation}
The unitary matrix $U$ for the entropy converter, which transforms the coherence matrices according to $\mathbf{G'}\!=\!{U}\mathbf{G}{U}^{\dag}$, is given by:

\begin{flalign}\label{Eq:transform_U}&U=U_{\mathrm{W}}(\theta_{2})U_{\mathrm{PBS}}U_{\mathrm{W}}(\theta_{1})\nonumber
\\
&=\left(\begin{array}{cccc}
\sin2\theta_{1}\sin2\theta_{2}&-\cos2\theta_{1}\cos2\theta_{2}&\cos2\theta_{2}&0\\
-\sin2\theta_{1}\cos2\theta_{2}&\cos2\theta_{1}\cos2\theta_{2}&\sin2\theta_{2}&0\\\cos2\theta_{1}&\sin2\theta_{1}&0&0\\0&0&0&1
\end{array}\right)
.
\end{flalign}    

When we set $\theta_{1}\!=\!45^{\circ}$, this matrix becomes:
\begin{equation}
U=\left(\begin{array}{cccc}
\sin2\theta_{2}&0&\cos2\theta_{2}&0\\
-\cos2\theta_{2}&0&\sin2\theta_{2}&0\\0&1&0&0\\0&0&0&1
\end{array}\right).
\end{equation}
The two settings of interest in our experiments are associated with $\theta_{2}\!=\!0^{\circ}$,
\begin{equation}\label{Eq:transform1U}
U=\left(\begin{array}{cccc}
0&0&1&0\\-1&0&0&0\\0&1&0&0\\0&0&0&1
\end{array}\right),
\end{equation}
and that associated with $\theta_{2}\!=\!45^{\circ}$
\begin{equation}\label{Eq:transform2_U}
U=\left(\begin{array}{cccc}
1&0&0&0\\0&0&1&0\\0&1&0&0\\0&0&0&1
\end{array}\right).
\end{equation}

\subsection{Theoretical coherence matrices}

We provide here the theoretical expectations for the coherence matrices corresponding to the fields in Fig.~3 of the main text. In general, the setting $\theta_{2}\!=\!0^{\circ}$ produces the following transformation to a diagonal matrix:
\begin{equation}
\mathrm{diag}\{\lambda_{1},\lambda_{2},\lambda_{3},\lambda_{4}\}\rightarrow\mathrm{diag}\{\lambda_{3},\lambda_{1},\lambda_{2},\lambda_{4}\}.
\end{equation}
For the rank-2 fields, this corresponds to:
\begin{equation}
\mathrm{diag}\{\lambda_{1},0,\lambda_{2},0\}\rightarrow\mathrm{diag}\{\lambda_{2},\lambda_{1},0,0\};
\end{equation}
and for rank-3 fields:
\begin{equation}
\mathrm{diag}\{\lambda_{1},0,\lambda_{2},\lambda_{3}\}\rightarrow\mathrm{diag}\{\lambda_{2},\lambda_{1},0,\lambda_{3}\}.
\end{equation}
Alternatively, the setting $\theta_{2}\!=\!45^{\circ}$ produces the following transformation to a diagonal matrix:
\begin{equation}
\mathrm{diag}\{\lambda_{1},\lambda_{2},\lambda_{3},\lambda_{4}\}\rightarrow\mathrm{diag}\{\lambda_{1},\lambda_{3},\lambda_{2},\lambda_{4}\}.
\end{equation}
For the rank-2 fields, this corresponds to:
\begin{equation}
\mathrm{diag}\{\lambda_{1},0,\lambda_{2},0\}\rightarrow\mathrm{diag}\{\lambda_{1},\lambda_{2},0,0\};
\end{equation}
and for rank-3 fields:
\begin{equation}
\mathrm{diag}\{\lambda_{1},0,\lambda_{2},\lambda_{3}\}\rightarrow\mathrm{diag}\{\lambda_{1},\lambda_{2},0,\lambda_{3}\}.
\end{equation}

The rank-2 field in Fig.~3(d) with targeted entropy $S\!=\!0.75$ has the coherence matrix:
\begin{equation}
\mathbf{G}=\mathrm{diag}\{0.7855,0,0.2145,0\},
\end{equation}
corresponding to $S\!=\!0.75$, $S_{\mathrm{s}}\!=\!0.75$, and $S_{\mathrm{p}}\!=\!0$. With $\theta_{2}\!=\!0^{\circ}$ for W$_2$, the coherence matrix becomes:
\begin{equation}
\mathbf{G}'=\mathrm{diag}\{0.2145,0.7855,0,0\},
\end{equation}
corresponding to $S\!=\!0.75$, $S_{\mathrm{s}}\!=\!0$, and $S_{\mathrm{p}}\!=\!0.75$.

After the entropy converter with $\theta_{2}\!=\!45^{\circ}$ for W$_2$, the coherence matrix becomes:
\begin{equation}
\mathbf{G}'=\mathrm{diag}\{0.7855,0.2145,0,0\},
\end{equation}
corresponding to $S\!=\!0.75$, $S_{\mathrm{s}}\!=\!0$, and $S_{\mathrm{p}}\!=\!0.75$.

For the rank-3 field in Fig.~3(d) with targeted entropy $S\!=\!0.75$, the field has the coherence matrix:
\begin{equation}
\mathbf{G}=\mathrm{diag}\{0.8528,0,0.0736,0.0736\},
\end{equation}
corresponding to $S\!=\!0.75$, $S_{\mathrm{s}}\!=\!0.60$, and $S_{\mathrm{p}}\!=\!0.38$. With $\theta_{2}\!=\!0^{\circ}$ for W$_2$, the coherence matrix becomes:
\begin{equation}
\mathbf{G}'=\mathrm{diag}\{0.0736,0.8528,0,0.0736\},
\end{equation}
corresponding to $S\!=\!0.75$, $S_{\mathrm{s}}\!=\!0.38$, and $S_{\mathrm{p}}\!=\!0.38$.

After the entropy converter with $\theta_{2}\!=\!45^{\circ}$ for W$_2$, the coherence matrix becomes:
\begin{equation}
\mathbf{G}'=\mathrm{diag}\{0.8528,0.0736,0,0.0736\},
\end{equation}
corresponding to $S\!=\!0.75$, $S_{\mathrm{s}}\!=\!0.38$, and $S_{\mathrm{p}}\!=\!0.60$.

The rank-2 field in Fig.~3(e) of the main text with targeted entropy $S\!=\!1$ has the coherence matrix:
\begin{equation}
\mathbf{G}=\mathrm{diag}\{0.5,0,0.5,0\},
\end{equation}
corresponding to $S\!=\!1$, $S_{\mathrm{s}}\!=\!1$, and $S_{\mathrm{p}}\!=\!0$. After the entropy converter with $\theta_{2}\!=\!0^{\circ}$ for W$_2$, the coherence matrix becomes:
\begin{equation}
\mathbf{G}'=\mathrm{diag}\{0.5,0.5,0,0\},
\end{equation}
corresponding to $S\!=\!1$, $S_{\mathrm{s}}\!=\!0$, and $S_{\mathrm{p}}\!=\!1$. With $\theta_{2}\!=\!45^{\circ}$ for W$_2$, the coherence matrix becomes:
\begin{equation}
\mathbf{G}'=\mathrm{diag}\{0.5,0.5,0,0\},
\end{equation}
corresponding to $S\!=\!1$, $S_{\mathrm{s}}\!=\!0$, and $S_{\mathrm{p}}\!=\!1$.

The rank-3 field in Fig.~3(e) of the main text with targeted entropy $S\!=\!1$ has the coherence matrix:
\begin{equation}
\mathbf{G}=\mathrm{diag}\{0.7730,0,0.1135,0.1135\},
\end{equation}
corresponding to $S\!=\!1$, $S_{\mathrm{s}}\!=\!0.77$, and $S_{\mathrm{p}}\!=\!0.51$. After the entropy converter with $\theta_{2}\!=\!0^{\circ}$ for W$_2$, the coherence matrix becomes:
\begin{equation}
\mathbf{G}'=\mathrm{diag}\{0.1135,0.7730,0,0.1135\},
\end{equation}
corresponding to $S\!=\!1$, $S_{\mathrm{s}}\!=\!0.51$, and $S_{\mathrm{p}}\!=\!0.51$. With $\theta_{2}\!=\!45^{\circ}$ for W$_2$, the coherence matrix becomes:
\begin{equation}
\mathbf{G}'=\mathrm{diag}\{0.7730,0.1135,0,0.1135\},
\end{equation}
corresponding to $S\!=\!1$, $S_{\mathrm{s}}\!=\!0.51$, and $S_{\mathrm{p}}\!=\!0.77$.

Finally, the rank-3 field in Fig.~3(f) of the main text has the coherence matrix:
\begin{equation}
\mathbf{G}=\frac{1}{3}\mathrm{diag}\{1,0,1,1\},
\end{equation}
corresponding to $S\!=\!1.585$, $S_{\mathrm{s}}\!=\!0.92$, and $S_{\mathrm{p}}\!=\!0.92$. With $\theta_{2}\!=\!0^{\circ}$ for W$_2$, the coherence matrix becomes:
\begin{equation}
\mathbf{G}'=\frac{1}{3}\mathrm{diag}\{1,1,0,1\},
\end{equation}
corresponding to $S\!=\!1.585$, $S_{\mathrm{s}}\!=\!0.92$, and $S_{\mathrm{p}}\!=\!0.92$.

 After the entropy converter with $\theta_{2}\!=\!45^{\circ}$ for W$_2$, the coherence matrix becomes:
\begin{equation}
\mathbf{G}'=\frac{1}{3}\mathrm{diag}\{1,1,0,1\},
\end{equation}
corresponding to $S\!=\!1.585$, $S_{\mathrm{s}}\!=\!0.92$, and $S_{\mathrm{p}}\!=\!0.92$.

\subsection{Theoretical coherence matrices for the double-slit interference measurements}

We present here the coherence matrices associated with the fields used in Fig.~4 of the main text.

The rank-2 field in Fig.~4(a) with $S\!=\!0.75$ has the coherence matrix:
\begin{equation}
\mathbf{G}=\left(\begin{array}{cccc}
0.3927&0.3927&0&0\\
0.3927&0.3927&0&0\\
0&0&0.1073&0.1073\\
0&0&0.1073&0.1073
\end{array}\right),
\end{equation}
corresponding to $S\!=\!0.75$, $S_{\mathrm{s}}\!=\!0.75$, and $S_{\mathrm{p}}\!=\!0$. To realize these values, the power at $b$ is adjusted via a neutral density filter to be $\approx\!23\%$ that of the power at $a$. After the entropy converter with W$_2$ oriented at $0^{\circ}$ is:
\begin{equation}
\mathbf{G}'=\left(\begin{array}{cccc}
0.1073&0&0&0.1073\\
0&0.3927&0.3927&0\\
0&0.3927&0.3927&0\\
0.1073&0&0&0.1073
\end{array}\right),
\end{equation}
corresponding to $S\!=\!0.75$, $S_{\mathrm{s}}\!=\!1$, and $S_{\mathrm{p}}\!=\!1$. When W$_{2}$ is oriented at $45^{\circ}$, the coherence matrix is:
\begin{equation}
\mathbf{G}'=\left(\begin{array}{cccc}
0.3927&0&0.3927&0\\
0&0.1073&0&0.1073\\
0.3927&0&0.3927&0\\
0&0.1073&0&0.1073
\end{array}\right),
\end{equation}
corresponding to $S\!=\!0.75$, $S_{\mathrm{s}}\!=\!0$, and $S_{\mathrm{p}}\!=\!0.75$.

The rank-3 field in Fig.~4(a) is produced after adjusting the power at $b$ to be $\approx\!15\%$ that at $a$, and is associated with the coherence matrix:
\begin{equation}
\mathbf{G}=\left(\begin{array}{cccc}
0.4264&0.4264&0&0\\
0.4264&0.4264&0&0\\
0&0&0.0736&0\\
0&0&0&0.0736
\end{array}\right),
\end{equation}
corresponding to $S\!=\!0.75$, $S_{\mathrm{s}}\!=\!0.60$, and $S_{\mathrm{p}}\!=\!0.38$. After the entropy converter with W$_2$ oriented at $0^{\circ}$, the coherence matrix is:
\begin{equation}
\mathbf{G}'=\left(\begin{array}{cccc}
0.0736&0&0&0\\
0&0.4264&0.4264&0\\
0&0.4264&0.4264&0\\
0&0&0&0.0736
\end{array}\right),
\end{equation}
corresponding to $S\!=\!0.75$, $S_{\mathrm{s}}\!=\!1$, and $S_{\mathrm{p}}\!=\!1$. When W$_2$ is oriented at $45^{\circ}$, the coherence matrix is:
\begin{equation}
\mathbf{G}'=\left(\begin{array}{cccc}
0.4264&0&0.4264&0\\
0&0.0736&0&0\\
0.4264&0&0.4264&0\\
0&0&0&0.0736
\end{array}\right),
\end{equation}
corresponding to $S\!=\!0.75$, $S_{\mathrm{s}}\!=\!0.38$, and $S_{\mathrm{p}}\!=\!0.60$.

The rank-2 field in Fig.~4(b) with $S\!=\!1$ is prepared with equal power at $a$ and $b$, and is associated with the coherence matrix: 
\begin{equation}
\mathbf{G}=\frac{1}{4}\left(\begin{array}{cccc}
1&1&0&0\\
1&1&0&0\\
0&0&1&1\\
0&0&1&1
\end{array}\right),
\end{equation}
corresponding to $S\!=\!1$, $S_{\mathrm{s}}\!=\!1$, and $S_{\mathrm{p}}\!=\!0$. After the entropy converter with W$_2$ oriented at $0^{\circ}$, the coherence matrix is:
\begin{equation}
\mathbf{G}'=\frac{1}{4}\left(\begin{array}{cccc}
1&0&0&1\\
0&1&1&0\\
0&1&1&0\\
1&0&0&1
\end{array}\right),
\end{equation}
corresponding to $S\!=\!1$, $S_{\mathrm{s}}\!=\!1$, and $S_{\mathrm{p}}\!=\!1$. In this case $S\!=\!S_{\mathrm{s}}\!=\!S_{\mathrm{p}}$, and the reduced coherence matrices are: ${\mathbf{G}_\mathrm{s}\!=\!\mathbf{G}_\mathrm{p}\!=\!\begin{pmatrix}0.5&0\\0&0.5\\\end{pmatrix}}$. After the entropy converter with W$_2$ oriented at $45^{\circ}$, the coherence matrix is:
\begin{equation}
\mathbf{G}'=\frac{1}{4}\left(\!\begin{array}{cccc}
1&0&1&0\\
0&1&0&1\\
1&0&1&0\\
0&1&0&1
\end{array}\right),
\end{equation}
corresponding to $S\!=\!1$, $S_{\mathrm{s}}\!=\!0$, and $S_{\mathrm{p}}\!=\!1$.

The rank-3 field in Fig.~4(b) is prepared by adjusting the power at $b$ to be $\approx\!33\%$ of that at $a$, and is associated with the coherence matrix:
\begin{equation}
\mathbf{G}=\left(\begin{array}{cccc}
0.3865&0.3865&0&0\\
0.3865&0.3865&0&0\\
0&0&0.1135&0\\
0&0&0&0.1135
\end{array}\right),
\end{equation}
corresponding to $S\!=\!1$, $S_{\mathrm{s}}\!=\!0.77$, and $S_{\mathrm{p}}\!=\!0.51$. After the entropy converter with W$_2$ oriented at $0^{\circ}$, the coherence matrix is:
\begin{equation}
\mathbf{G}'=\left(\begin{array}{cccc}
0.1135&0&0&0\\
0&0.3865&0.3865&0\\
0&0.3865&0.3865&0\\
0&0&0&0.1135
\end{array}\right),
\end{equation}
corresponding to $S\!=\!1$, $S_{\mathrm{s}}\!=\!1$, and $S_{\mathrm{p}}\!=\!1$. In this case, $S\!=\!S_{\mathrm{s}}\!=\!S_{\mathrm{p}}$, with the reduced coherence matrices $\mathbf{G}_\mathrm{s}\!=\!\mathbf{G}_\mathrm{p}\!=\!\begin{pmatrix}0.5&0\\0&0.5\\\end{pmatrix}$. With W$_2$ oriented at $45^{\circ}$, the coherence matrix is:
\begin{equation}
\mathbf{G}'=\left(\begin{array}{cccc}
0.3865&0&0.3865&0\\
0&0.1135&0&0\\
0.3865&0&0.3865&0\\
0&0&0&0.1135
\end{array}\right),
\end{equation}
corresponding to $S\!=\!1$, $S_{\mathrm{s}}\!=\!0.51$, and $S_{\mathrm{p}}\!=\!0.77$. 

\section{Theoretical results}

We prove here the two theorems described in the main text in which the rank of $\mathbf{G}$ is related to polarization projections that yield a spatially coherent scalar field \cite{Kelkar:19}.

\begin{theorem}\label{thm:rank3}
For a vector optical field supported on two spatial points with a coherence matrix $\mathbf{G}$, if there exists a polarization projection along vector $\mathbf{P}$ along which the field is spatially coherent (i.e., it can produce spatial interference fringes with $100\%$ visibility), then $\mathcal{R}(\mathbf{G})\leq3$. 
\end{theorem}

\begin{proof}

We consider the polarization vector $\mathbf{P}\!=\!\left(\begin{array}{c}\cos\theta\\e^{i\phi}\sin\theta\end{array}\right)$ with $\theta\in[0,\pi/2]$ and $\phi\in[0,2\pi)$, and take $\mathbf{Q}$ to be the orthogonal polarization vector. The block-diagonal unitary transformation $U$ given by:
\begin{equation}
U=I\otimes\begin{pmatrix}
\cos\theta & e^{i\phi}\sin\theta\\e^{-i\phi}\sin\theta & -\cos\theta\end{pmatrix},
\end{equation}
transforms the field from the polarization basis defined by H and V to that defined by $\mathbf{Q}$ and $\mathbf{P}$, $\mathbf{u}_{\mathrm{PQ}}\!=\!U\mathbf{u}_{\mathrm{HV}}$;
here $\mathbf{u}_{\mathrm{HV}}\!=\![E_a^{\mathrm{H}} ~~ E_a^{\mathrm{V}} ~~ E_b^{\mathrm{H}} ~~ E_b^{\mathrm{V}}]^\top$ and $\mathbf{u}_{\mathrm{PQ}} = [E_a^{\mathrm{P}} ~~ E_a^{\mathrm{Q}} ~~ E_b^{\mathrm{P}} ~~ E_b^{\mathrm{Q}}]^\top$ are vector fields defined in these two polarization bases.

The coherence matrix in the PQ polarization basis is:
\begin{eqnarray}
\boldsymbol{\Gamma} &=&\langle \mathbf{u}_{PQ}\mathbf{u}_{PQ}^\dagger \rangle = U\langle\mathbf{u}_{HV}\mathbf{u}_{HV}^\dagger \rangle U^\dagger=U\mathbf{G}U^\dagger\nonumber\\
&=& \left(\begin{array}{cccc}
\Gamma_{aa}^{PP}&\Gamma_{aa}^{PQ}&\Gamma_{ab}^{PP}&\Gamma_{ab}^{PQ}\\
\Gamma_{aa}^{QP}&\Gamma_{aa}^{QQ}&\Gamma_{ab}^{QP}&\Gamma_{ab}^{QQ}\\
\Gamma_{ba}^{PP}&\Gamma_{ba}^{PQ}&\Gamma_{bb}^{PP}&\Gamma_{bb}^{PQ}\\
\Gamma_{ba}^{QP}&\Gamma_{ba}^{QQ}&\Gamma_{bb}^{QP}&\Gamma_{bb}^{QQ}\\\end{array}\right),\label{eqn:Gamma}
\end{eqnarray}
where $\Gamma_{kl}^{mn}=\langle E_{k}^{m}E_{l}^{n*}\rangle$, $m,n=P,Q$ and $k,l=a,b$, and $\langle\cdot\rangle$ denotes the ensemble average or expectation.

When $E_a^P$ and $E_b^P$ are mutually coherent, then there exists a complex constant $\alpha$ whereby $E_a^P=\alpha E_b^P$, which lead to: $\Gamma_{ka}^{mP}=\alpha^* \Gamma_{kb}^{mP}$, for $k = a,b$ and $m = P,Q$. This implies that columns 1 and 3 of $\boldsymbol{\Gamma}$ are linearly dependent. Therefore, $\mathcal{R}(\boldsymbol{\Gamma})\leq3$. Recall that $\boldsymbol{\Gamma}$ in the $PQ$ polarization basis is related to $\mathbf{G}$ in the HV basis via a unitary transformation $\boldsymbol{\Gamma} = U\mathbf{G}U^\dagger$ that does not change the rank. Therefore, we have $\mathcal{R}(\mathbf{G})\leq3$.
\end{proof}

\begin{theorem}\label{thm:rank2}
For a vector optical field supported on two points with a coherence matrix $\mathbf{G}$, if there exist two orthogonal polarization projections $\mathbf{P}$ and $\mathbf{Q}$ along which the field is spatially coherent (i.e., it can produce spatial interference fringes with $100\%$ visibility), then $\mathcal{R}(\mathbf{G})\leq2$. 
\end{theorem}

\begin{proof}
Similarly to the proof of Theorem~\ref{thm:rank3}, there exist complex-valued constants $\alpha$ and $\beta$, such that $E_a^P=\alpha E_b^P$ and $E_a^Q=\beta E_b^Q$, and the the coherence matrix $\Gamma$ in the PQ polarization basis is related by a unitary transformation $U$ to the coherence matrix $\mathbf{G}$ in the HV basis: $\Gamma = U\mathbf{G}U^\dagger$. Therefore, $\Gamma_{ka}^{mP}=\alpha^*\Gamma_{kb}^{mP}$ and $\Gamma_{ka}^{mQ}=\alpha^*\Gamma_{kb}^{mQ}$, for $k=a,b$ and $m=P,Q$, which implies that columns 1 and 3 of $\boldsymbol{\Gamma}$ are linearly dependent, and so are columns 2 and 4. Therefore we have $\mathcal{R}(\mathbf{G})=\mathcal{R}(\boldsymbol{\Gamma})\leq2$.

\end{proof}

Note, however, that the converses of these two theorems are not true. However, for the specific case of the theoretical matrix of a rank-3 field used in Fig.~4 of the main text, the converse of Theorem \ref{thm:rank3} holds. This is because the first and third columns of $\mathbf{G}_3'$ are equal. Therefore, for the optical field $\mathbf{E} = [ E_a^H \quad E_a^V \quad E_b^H \quad E_b^V ]^\top$ corresponding to $\mathbf{G}_3'$,
\begin{align}
    \langle E_a^H E_a^{H*} \rangle = \langle E_a^H E_b^{H*} \rangle = \langle E_b^H E_b^{H*} \rangle,
\end{align}
which implies that $E_a^H$ and $E_b^H$ are maximally correlated. We also have $\langle E_a^V E_b^{V*} \rangle = 0$, which implies that $E_a^V$ and $E_b^V$ are incoherent. Therefore, for the specific case of $\mathbf{G}_3'$, it is possible to find a polarization projection that renders the field at two points coherent. Additionally, the field at these two points is incoherent along the orthogonal polarization projection. Similarly, for the specific case of the rank-2 field used in Fig.~4 of the main text. 
It can be shown that the field is coherent along two orthogonal polarization projections, H and V.

\bibliography{supp}